\documentclass{article} %{ifacconf}
\pdfoutput=1
\usepackage{graphicx}      % include this line if your document contains figures
\usepackage{authblk}

\usepackage[cmex10]{amsmath}
\usepackage{amssymb,graphics,amscd,amsopn,amstext,amsfonts,amsthm}
\newtheorem{proposition}{Proposition}
\newtheorem{theorem}{Theorem}
\newtheorem{corollary}{Corollary}
\newtheorem{assumption}{Assumption}
\newtheorem{lemma}{Lemma}
\newtheorem{remark}{Remark}
\newtheorem{problem}{Problem}

\newcommand{\be}{\begin{equation}}
\newcommand{\ee}{\end{equation}}
\newcommand{\bea}{\begin{eqnarray}}
\newcommand{\eea}{\end{eqnarray}}
\newcommand{\beas}{\begin{eqnarray*}}
\newcommand{\eeas}{\end{eqnarray*}}
\newcommand{\nn}{\nonumber}
\newcommand{\bbm}{\begin{bmatrix}}
\newcommand{\ebm}{\end{bmatrix}}
\newcommand{\matl}{\left[ \begin{array}}
\newcommand{\matr}{\end{array} \right]}

\newcommand{\cC}{\mathcal{C}}
\newcommand{\cD}{\mathcal{D}}

\newcommand{\bR}{\mathbb{R}}
\newcommand{\bS}{\mathbb{S}}
\newcommand{\bW}{\mathbb{N}}

\newcommand{\cB}{\mathcal{B}}
\newcommand{\cK}{\mathcal{K}}
\newcommand{\cF}{\mathcal{F}}
\newcommand{\cG}{\mathcal{G}}
\newcommand{\cI}{\mathcal{I}}
\newcommand{\cM}{\mathcal{M}}

\newcommand{\Delt}{\Delta t}

\newcommand{\mrm}{\mathrm}

\newcommand{\T}{^{\mbox{\small T}}}

\begin{document}
%\begin{frontmatter}

\title{Nonlinearly Stable Real-Time Learning and Model-Free Control}%\thanksref{footnoteinfo}} 

%\thanks[footnoteinfo]{This research was sponsored by a Syracuse University internal grant.}

\author[$\dagger$]{A. K. Sanyal} 
\affil[$\dagger$]{Department of Mechanical and Aerospace Engineering, Syracuse University, Syracuse, 
NY 13104, ph: (315) 443--0466, email: aksanyal@syr.edu} 
%}
\maketitle
%\author[First]{Amit K. Sanyal} 
%
%\address[First]{Syracuse University, Mechanical and Aerospace Engineering,
%   Syracuse, NY 13244 USA (e-mail: aksanyal@syr.edu).}

\begin{abstract}
This work provides a framework for nonlinear model-free control of systems whose input-output 
dynamics are unknown or uncertain, with outputs that can be controlled by the inputs. This 
framework leads to real-time control of the system such that any feasible output trajectory can 
be tracked by the inputs. Unlike existing model-free or data-driven control approaches, the framework 
given here provides guaranteed nonlinear stability based on a Lyapunov stability analysis. The 
controller and observer designs in the proposed framework are nonlinearly finite-time stable and 
robust to the unknown or uncertain dynamics as well as unknown measurement noise. For ease 
of computer implementation, the framework is developed in discrete time. Nonlinear stability analysis 
of the discrete-time observers and controllers are carried out using discrete Lyapunov analysis. The 
unknown or uncertain input-output dynamics is learnt in real time using a nonlinearly stable observer. 
This observer ensures that the uncertain/unknown dynamics is learnt from prior input-output history 
and if the uncertainty in the model is bounded, then the error in the estimate of this dynamics is 
also bounded. Moreover, this observer ensures finite-time stable convergence of model estimation 
errors to zero if the unknown model is constant (not time-varying), and model estimation errors  
converge to a bounded neighborhood of the zero vector if the rate of change of the model is bounded. 
Output measurements are filtered by a finite-time stable observer before being used for feedback 
tracking of a desired output trajectory. Finite-time stable observer design in this framework also 
ensures that a nonlinear separation principle is in effect for separate controller and observer design. 
A model-free nonlinearly stable control scheme is then designed to ensure convergence of observed 
outputs to a desired output trajectory. This control scheme ensures nonlinear finite-time stable 
convergence of tracking errors to a manifold where the tracking errors decay asymptotically. A 
numerical experiment on a nonlinear second-order system demonstrates the performance of this 
nonlinear model-free control framework. %266 words 
\end{abstract}

%\begin{keyword}
%Model-free control, nonlinear control, second-order systems.
%\end{keyword}
%
%\end{frontmatter}
%===============================================================================

\section{Introduction}\label{sec: intro}
For feedback control of nonlinear systems with uncertain or unknown input-output dynamics, 
data-driven control approaches have been proposed and used. When only input-output behavior 
of the system is available, then output regulation to a desired set point or output trajectory tracking 
has to be based on model-free (data-driven) controller and observer designs. This work provides
nonlinear model-free controller and observer designs for output tracking of systems for which only 
input-output knowledge is available. This framework for nonlinear model-free control of second-order 
systems can be implemented on systems of any order, as described in this paper. This framework is 
applicable to nonlinear systems for which output measurements are available and these outputs are 
controllable with the applied inputs. The main contribution of this work is that provides definite 
(quantifiable) guarantees on nonlinear stability of output tracking control and robustness to 
uncertainties, while learning the local input-output behavior of the system in real time. The 
development and implementation of this framework is carried out in discrete time for ease of 
computer implementation.

A majority of linear and nonlinear control approaches are model-based, for which a model of the 
dynamics of the system being controlled is necessary. However, as the number and variety of applications 
of systems and control theory continues to increase, uncertainties and difficulties in modeling systems are 
becoming increasingly important. Of particular interest for this work is the large class of 
(nonlinear) systems with uncertain or unmodeled dynamics that need to be controlled in real time.
This class of systems includes, for example, autonomous vehicles, walking robots, and electronic medical 
implants. For such systems, ``model-free" (i.e., data-driven) control techniques may be used for feedback 
control in real-time. In the last 15 years, the term ``model-free control" for uncertain systems has been 
used in different senses and settings in the published literature. 
%, as remarked by Michel Fliess in~\cite{fljosi08,FlJo13}. 
These settings are quite varied, and range from ``classic" PIDs to feedback control using techniques 
from neural nets, fuzzy logic, and soft computing to learn the uncertainties in the dynamics, e.g.,
in~\cite{kelbha08,kilkrst06,doco10,sata11,renbig17}. A linear model-free control framework, 
termed the ``intelligent PID" (or ``iPID") scheme, was proposed in~\cite{fljosi08,FlJo13}. 
The iPID framework uses a linear {\em ultra-local model} to describe the unknown input-output dynamics, 
and estimates and uses this model for feedback control. In addition, if the system is known to be 
differentially flat for the selected outputs~\cite{fllemaro95}, then a state trajectory can also be tracked 
and uncertainties in the input-state dynamics can also be estimated over time from the measured outputs 
using model-free filtering techniques~\cite{fljosi08,trsiba07}. However, in the iPID framework, the 
ultra-local model is estimated by a linear filtering scheme assuming measurements at a sufficiently high 
sampling frequency, as reiterated in~\cite{griz17}. Moreover, there are no accompanying stability 
guarantees provided for this iPID scheme. More recently, a data-enabled predictive control (``DeePC") 
method was formulated for data-driven control of unmodeled/uncertain systems that is analogous to the 
classical model-predictive control (MPC) technique for model-based control of linear systems 
in~\cite{DeePC}. Some applications where model-free control techniques have been used are 
given in, e.g.,~\cite{machines-mfc,ydn16,viba11,chgagu11}. While these model-free or data-driven 
techniques may often prove to be useful engineering fixes, they may not converge to desired outputs or 
states in a nonlinearly stable manner. In fact, guaranteed nonlinear stability in the sense of \cite{lyapthes} 
is notably missing in these approaches. 

While prior work using the iPID framework has continuous-time feedback control (e.g., 
in~\cite{fljosi08,FlJo13,machines-mfc,ydn16}), the framework here uses discrete-time nonlinear 
model-free estimation and control for tracking desired output trajectories. 
The framework given here lays the foundation for nonlinearly stable model-free control, using novel 
methods to estimate the local input-output model and filter out noise from measured outputs. Our past 
research using H\"{o}lder-continuous finite-time stable control and estimation schemes in continuous
time have appeared in, e.g.,~\cite{bosa_ijc,sasi:2017:finite,ecc19,asjc19}. 
%and past input-output data without violating causality. 
The finite-time stability of the estimation schemes provides a natural separation of the estimation 
process from the tracking control, as the estimators can be designed to converge in a time that is smaller 
than the settling time of the controller. For tracking a desired output or state trajectory, a 
H\"{o}lder-continuous nonlinear finite-time stable (FTS) tracking control scheme was used 
in~\cite{bosa_ijc,sasi:2017:finite}. These continuous FTS schemes are based on the Lyapunov 
analysis in~\cite{bhatFT}, using H\"{o}lder-continuous Lyapunov functions. Here, we develop a basic 
result that extends this Lyapunov analysis to provide finite-time stable convergence in discrete time. 
In addition, {\em this framework provides guaranteed nonlinear stability of the overall feedback 
system without requiring high frequencies for measurement or control}. The overall emphasis in 
our approach is towards {\em guaranteeing nonlinear feedback system stability} over every other 
system-theoretic property (including ``near optimality"). 

We provide the schemes in discrete-time as they are easier to implement numerically and experimentally. 
In the first part of this framework, a nonlinear, finite-time stable, model-free, estimation scheme is 
designed to estimate the outputs from noisy measurements. In the second part of our nonlinear 
model-free control framework, a nonlinear finite-time stable observer is designed to predict the 
unknown ultra-local model describing the local input-output relation based on observed input-output 
behavior at prior instants. This is a critical component of our framework, as it ensures nonlinear  
stability of the overall feedback loop. The nonlinear finite-time stable observers designed in the first 
two components of our framework ensure that a separation principle is in effect for nonlinear observer 
design independent of the control design. In the third and last part of this framework, we design a 
nonlinearly stable, trajectory tracking control scheme designed to track a desired output trajectory. 
A nonlinearly stable control scheme is designed that ensures convergence to a desired manifold in 
the output space in finite time, and this manifold is designed to ensure that the output tracking 
error converges to zero exponentially. %This control scheme is also designed in discrete time for 
%computer implementation, although corresponding continuous time versions can be obtained 
%immediately using the principles provided here.  

The remainder of this paper is organized as follows. The mathematical formulation for model-free control 
of a nonlinear system is introduced along with preliminary results on finite-time stability in discrete time,
in Section \ref{sec: nonsysa}. Section \ref{sec: ftsobs} provides a finite-time stable observer design for 
estimating the outputs and filtering out measurement noise from output measurements. Output estimates
from this finite-time stable observer are used with the feedback tracking control scheme. In section 
\ref{sec: ulm}, a finite-time stable observer is designed to estimate an ultra-local model relating the 
inputs and outputs of the nonlinear system. This ultra-local model also depends on the assumed order 
of the system (which may be unknown), and an assumed {\em influence matrix} relating the derivative of 
the output vector of this order to the control input vector (which is designed as part of the control 
framework). The model-free control law for output tracking is given in Section \ref{sec: mfccon}. This 
section provides a discrete time control law that makes the output tracking error converge to a manifold
in a finite-time stable manner, and this manifold is designed such that the output tracking error converges 
to zero asymptotically. Section \ref{sec: numres} provides numerical simulation results of applying this 
nonlinear model-free control framework to output trajectory tracking of a second order system: the 
inverted pendulum on a cart. The model of this system is assumed to be unknown for purposes of 
control design, while the simulated system has nonlinear friction terms affecting the dynamics of both 
the degrees of freedom. The results of this numerical experiment applying this framework to this 
well-known unstable system corroborate its analytical stability and robustness properties. Finally, 
section \ref{sec: conc} provides a summary of the main results of this paper and ends with our 
planned future work in this area.
%As mentioned previously, we only provide the discrete-time 
%controller, but the development provides sufficient hints on how to obtain the continuous-time 
%version using the same principles used to design the discrete-time controller. 

\section{Nonlinear system assumptions}\label{sec: nonsysa}
All notation used in this paper is defined on first use, and unnecessary notation is not used.
Consider a $\nu$'th-order (and relative degree) nonlinear system with $m$ inputs and $l$ outputs, where 
$m$ and $l$ are positive integers and $m\ge l$. %$n\ge l$. 
In continuous time, the dynamics model of the system relates inputs and output according to:
\be y^{(\nu)}(t)= \psi\big((y,\dot y,\ldots,y^{(\nu-1)},u)(t), t\big)
%\ddot y (t)\approx \frac{y_{k+1}-2y_k+y_{k-1}}{\Delt^2}= h(y_k,y_{k-1},u_k,t_k),  
\label{dynoutmod} \ee
where $\psi:(\bR^m)^\nu \times \bR^+ \to\bR^l$ is a continuous and possibly time-varying map, 
$(\cdot)^{(\mu)}$ denotes the $\mu$'th time derivative of a quantity, $y(t)$ is the system output and 
$u(t)$ is the control input at time $t$. This can be converted to a discrete time system where
$(\cdot)_k=(\cdot) (t_k)$ denotes the value of a time-varying quantity at sampling instant $t_k$, 
with $u_k\in\bR^m$ as the control input, 
%$\ddot y(t)$ is the second time derivative of $y(t)$ which is approximated 
%by the second order time difference, and $\Delt:=t_{k+1}- t_k$, 
where $k\in \bW=\{ 0,1,2,\ldots\}$ and $\bW$ denotes the index set of whole numbers including $0$. 
We consider the case that the the input-output model of the system \eqref{dynoutmod} (i.e., the function 
$\psi$) is unknown, but the outputs are measured and can be controlled by the inputs. This is stated in 
the following assumption.
\begin{assumption} \label{assump1}
The nonlinear system given by \eqref{dynoutmod} has $\psi(\cdots)$ unknown 
but outputs $y_k=y(t_k)$ are available from sensor measurements at sampling instants. Further, 
%The only assumption made for our nonlinear model-free control design is that 
this system is input-output controllable.  
\end{assumption}
To express the continuous time system \eqref{dynoutmod} in discrete time, we first replace 
derivatives in continuous time with finite differences in discrete time. In an abuse of notation,  
we use the superscript $(\mu)$ to denote the $\mu$th order finite difference of the output $y_k$ in 
place of the $\mu$th time derivative in eq. \eqref{dynoutmod}. The forward 
difference defined by
\be y_k^{(\mu)} := y_{k+1}^{(\mu-1)}-y_k^{(\mu-1)} \mbox{ with } y_k^{(0)}=y_k \label{fordif} \ee
is used, because of its simplicity and applicability for output tracking control. 
The control inputs $u_k$ are then designed so as to track a desired output trajectory $y^d_k=y^d(t_k)$ 
that is continuous and at least $\nu$ times differentiable, as described in Section \ref{sec: mfccon}. 

To make the discrete time output tracking control tractable, we make the additional assumption stated below.
\begin{assumption} \label{assump2}
The unknown input-output system \eqref{dynoutmod} can be expressed in discrete time as 
\be y_k^{(\nu)}= \varpi (y_k,y_k^{(1)},\ldots,y_k^{(\nu-1)},u_k,t_k), \label{discmod} \ee
where $\varpi$ is unknown but continuous, $\nu$ is known and the $\mu$th order finite difference 
$y_k^{(\mu)}$ is as defined by eq. \eqref{fordif}. Further, the desired output trajectory $y^d_k:= y^d(t_k)$ is continuous and $\nu$ times differentiable, and $(y_k^d)^{(\mu)}$ is bounded for $\mu\in\{1,\ldots,\nu\}$. 
\end{assumption}
When the order $\nu$ is unknown, then there are two options available: $\nu$ can be identified using known 
techniques (e.g.,~\cite{heas93, rhmo98}), or a sufficiently high order may be assumed for 
model-free control. 
%The degree of smoothness can be selected easily in most cases 
%where at least the order of the input-output model is known, or can be identified using techniques 
%available in the literature~\cite{heas93, rhmo98}. %\ref{sec: ftsobs}.
%but the measurement model is assumed to be 
%known at each instant. Let the measurement function be given by
%\be y(t)= h(x(t),t),\, \mbox{ where }\, h(\cdot,t): \bR^n\to \bR^p. \label{measmod} \ee
%\begin{remark}
%In addition to Assumption \ref{assump1}, if the system \eqref{dynoutmod} is differentially 
%flat~\cite{fllemaro95} and the output model is known, then a state trajectory $x^d(t)$ of the same order 
%of smoothness as the desired output trajectory $y^d (t)$, can also be tracked. 
%\end{remark}
%When the outputs $y(t)$ are flat, $x(t)$ can be reconstructed uniquely from $y(t)$ and its time 
%derivatives, which leads to the above remark. 

In practice, the outputs are measured by sensors that usually introduce noise that is modeled as 
additive noise:
\be y^m (t_k)= y^m_k= y_k+ \eta_k, \label{measnoise} \ee
where $\eta_k\in\bR^l$ is the additive noise. In order to filter out measurement noise from the measured 
output signals, we construct a finite-time stable observer in discrete time in Section \ref{sec: ftsobs}. This 
is the first step of the nonlinear model-free control framework formulated here. In the next step, a control  
affine ultra-local model (ULM) is constructed in discrete time and used to estimate the unknown 
dynamics from past input-output data in Section \ref{sec: ulm}. We design first and second order discrete 
time nonlinear observers that estimate the ultra-local model with finite-time stable convergence assuming 
that the uncontrolled part of the ULM is constant. These observers are also shown to be robust to bounded 
rates of change of the uncontrolled part of the ULM for the input-output dynamics. The final part of the 
framework given in Section \ref{sec: mfccon} uses the ULM, along with the ULM observer and the output 
observer, to construct an output feedback control scheme to track a desired output trajectory. 
%Denote the reconstructed state by
%\be x (t)= \eta \big(y(t),\dot y(t),\ldots, y^{(\nu)}(t),t\big), \label{xrecon} \ee
%where $\nu\in\mathbb{N}$ is the highest order (finite) time derivative of $y$ that is used to reconstruct 
%$x$. With the assumption that the measurement function is known, the function $\eta: \bR^{p\nu}\to 
%\bR^n$ is known. The measurement model relating the outputs to the states, and the dynamics model 
%relating the states to the inputs, are discretized in the following two sections for computer 
%implementation.

\section{Model-free finite-time stable observer} \label{sec: ftsobs}

To filter out initialization errors and measurement noise from measured output signals $y^m_k$ as given 
by eq. \eqref{measnoise}, we design a finite-time stable observer that gives robust and stable output 
estimates for output feedback control. In this work, the finite-time stable observer is not discontinuous but 
not Lipschitz continuous either; it is H\"{o}lder continuous. It filters out noise (of unknown statistics) from 
measurements, and in the absence of measurement noise provides finite-time stable convergence of 
output estimates to true outputs. Therefore, this observer can be used to provide continuous output 
feedback for output stabilization or tracking control. The primary benefits of finite-time stable observers in 
our framework for model-free control are two-fold: (1) the added robustness of finite-time stability compared 
to asymptotic stability for nonlinear systems when faced with the same bounds on intermittent or persistent 
disturbances~\cite{bhatFT, bosa_ijc}; and (2) convergence to zero errors in finite time conveniently ensures 
a separation principle is in effect for separate observer and controller designs. We design the finite-time 
stable (FTS) observer in discrete time so that it is suitable for numerical and embedded computer 
implementation. We define the output estimate error in discrete time as
%\be x_k = \eta_k (Y_k)\, \mbox{ where }\, Y_k= \big(y(t_k),\dot y(t_k),\ldots, y^{(\nu)}(t_k),t_k\big). 
\be e^o_k= e^o(t_k)= \hat y_k - y_k,\; k\in\cI. \label{measmodid} \ee 
%Denote the state estimation error in time instant $t_k$ by 
%\be e^o_k = \hat x_k - x_k = \hat x_k - \eta_k (Y_k), \label{errobs} \ee
%where $x_k= x(t_k)$, $y_k =h (x_k,t_k)$, $\hat x_k =\hat x(t_k)$ is the state estimate at time $t_k$, 
%and 
%\be \eta_ k := h(\cdot,t_k )^{-1}\, \mbox{ so that }\, \eta_k (Y_k)= x_k. \label{etadef} \ee 

The remainder of this section gives a finite-time stable (FTS) observer design in discrete time. The first 
subsection is a basic result on finite-time stability and convergence for discrete-time systems that, to the 
best of our knowledge, has not appeared in past research publications. The second result gives the 
finite-time stable observer design for our nonlinear model-free control framework.

\subsection{Finite-time stability in discrete time}\label{ssec: ftsdisc}

\begin{lemma}\label{discFTSlem}
Consider a discrete-time system with inputs $u_k\in\bR^m$ and outputs $y_k\in\bR^l$. Define a 
corresponding positive definite (Lyapunov) function $V:\bR^l \to \bR$ and let $V_k = V (y_k)$. Let 
$\alpha, \varepsilon$ be constants in the open interval $]0,1[$, let $V_0>0$ 
be the (finite) initial value of the Lyapunov function along an output trajectory $y_k$, and let $\gamma_k:= 
\gamma (V_k)$ where $\gamma: \bR^+ \to\bR^+$ is a positive definite function of $V_k$ %class-$\cK$ 
that satisfies  
\be %\gamma (\chi) \ge \chi \label{gammacond} \ee   \Big(\frac{V_k}{V_0}\Big)^{1-\alpha}
%\gamma_k =0 \Leftrightarrow V_k=0\, \mbox{ and }\, 
\frac{\gamma_k}{\gamma_0} \ge 1-\varepsilon \mbox{ for } V_k \in\, ] \chi V_0, V_0 [
%\gamma_k \ge V_k^{1-\alpha} \mbox{ for some finite } V_k.
\label{gammacond} \ee
for some (arbitrarily small) positive constant $0<\chi\ll 1$. Then, if $V_k$ satisfies the relation 
\be V_{k+1}- V_k \le -\gamma_k V_k^\alpha, \label{discFTS} \ee
the system is (Lyapunov) stable and $y_k$ converges to $y=0$ for $k\ge N$, for a finite integer $N\in\bW$.
\end{lemma}
\begin{proof}
Note that inequality \eqref{discFTS} is a sufficient condition for (Lyapunov) stability of the system, as it ensures
that the difference $V_{k+1}- V_k$ along trajectories of the discrete-time system is negative definite. It is 
sufficient to consider the equality case of \eqref{discFTS}, with the right-hand side of the equality being zero 
if and only if $V_k=0$, according to the definition of $\gamma_k$. This equality can be expressed as:
\begin{align} 
\begin{split}
V_{k+1} &= V_k -\gamma_k V_k^\alpha \\
&= V_k\Big( 1- \frac{\gamma_k}{V_k^{1-\alpha}}\Big). 
\end{split}\label{Vkp1} 
\end{align}
%Let $N\in\bW$ be such that inequality \eqref{gammacond} is satisfied for $V_N^{1-\alpha}$, i.e., 
%\[ \gamma_N := \gamma (V_N^{1-\alpha}) \ge V_N^{1-\alpha}. \] 
%Then from eq. \eqref{discFTS} we have:
%\[ V_{N+1}- V_N = - \gamma_N V_N^\alpha \le - V_N^{1-\alpha} V_N^\alpha = - V_N, \]
%which implies that $V_{N+1} \le 0$. As $V_{N+1}$ cannot be negative, it has to be zero. Consequently, 
%from eq. \eqref{discFTS} we see that $V_j=0$ for $j\ge N+1$. Therefore, $V_k$ converges to zero, and 
%as a result $y_k$ converges to zero, for $k>N$. 
%%%%%%%%%%%%%%%%%%%%%%%%%%%%%%%%%%%%%%%%%%%%%%%%%
Consider an arbitrary trajectory $y_k\in\bR^l$ of the discrete-time system. Let the initial value 
of the Lyapunov function along this trajectory be 
\be V_0= c_0 \big(\gamma_0\big)^{\frac1{1-\alpha}}, \mbox{ where } c_0>0. \label{initLyapf} \ee
Note that for any finite positive value of $V_0$, there exists an unique positive scalar $c_0$ that satisfies 
\eqref{initLyapf}. Substituting this value for $V_0$ in expression \eqref{discFTS}, we obtain:
\be \begin{split} 
V_1 - c_0 \big(\gamma_0\big)^{\frac1{1-\alpha}} &= -\gamma_0 c_0^\alpha  \big(\gamma_0
\big)^{\frac{\alpha}{1-\alpha}} \\
&= -c_0^\alpha \big(\gamma_0\big)^{\frac1{1-\alpha}} \\
\Rightarrow V_1 &= (c_0 -c_0^\alpha) \big(\gamma_0\big)^{\frac1{1-\alpha}}.
\end{split} \label{V1eqn}
\ee
Defining 
\[ c_1 := c_0 -c_0^\alpha, \]
equation \eqref{V1eqn} can be expressed as 
\[ V_1 = c_1 \big(\gamma_0\big)^{\frac1{1-\alpha}}. \]
Note that if $c_0 \le 1$, then the above implies that $c_1\le 0$, which leads to a contradiction unless 
$c_1=0$, as $V_1$ has to be non-negative from the definition of a Lyapunov function. In this case, the 
value of the Lyapunov function already converges to zero in the first step, i.e., for $N=1$. Now suppose 
$c_0>1$. In that case, substituting the above value for $V_1$ in \eqref{discFTS}, one obtains a similar 
expression for $V_2$:
\be V_2 = c_2 \big(\gamma_0\big)^{\frac1{1-\alpha}}\, \mbox{ where }\, c_2 := c_1 - a_1 c_1^\alpha 
\mbox{ and }  a_1 :=\frac{\gamma_1}{\gamma_0}. \label{V2exp} \ee
Continuing in this manner, we get the following expression for $V_{k+1}$ along with a recursive relation for 
the $c_k$ involving the $a_k$: 
\begin{align} 
\begin{split}
&V_{k+1} = c_{k+1} \big(\gamma_0\big)^{\frac1{1-\alpha}} \mbox{ for } k\ge 1, \mbox{ where }  \\
& c_{k+1} := c_k - a_k c_k^\alpha \mbox{ and }  a_k :=\frac{\gamma_k}{\gamma_0}. 
\end{split} \label{Vkp1} 
\end{align}
If $V_k$ is in the range given by \eqref{gammacond}, then according to eq. \eqref{Vkp1} and the 
inequality in \eqref{gammacond}, we have
\begin{align} 
\begin{split}
c_{k+1} &\le c_k - (1-\varepsilon) c_k^\alpha \\
&= \varepsilon c_k^\alpha - (1- c_k^{1-\alpha}) c_k^\alpha.
\end{split} \label{ckp1} 
\end{align}
As $V_{k+1}:=V(y_{k+1})$ is positive definite, $c_{k+1}$ cannot be negative according to eq. \eqref{Vkp1}. 
Further, if $V_k$ is in the range given by \eqref{gammacond}, we have:
\[ \frac{c_k}{c_0}= \frac{V_k}{V_0} \in ]\chi, 1[, \]
where $\chi$ is arbitrarily small; in particular, for $\chi c_0 <1$. From the right side of the inequality 
\eqref{ckp1}, we see that 
\be c_{k+1}\le 0\, \Leftrightarrow\, \varepsilon \le 1- c_k^{1-\alpha}\, \Leftrightarrow\, c_k^{1-\alpha}\le 
1-\varepsilon. \label{ck1mal} \ee
As $c_k\to \chi c_0 <1$ in this interval of $V_k$, there is a finite integer $k=N-1$ for which the inequality in 
\eqref{ck1mal} is satisfied, i.e., $c_{N-1}\le (1-\varepsilon)^{\frac1{1-\alpha}}$; and thus $c_N\le 0$. But 
$c_N\ge 0$ because $V_N\ge 0$ and $\gamma_0>0$. This leads to the conclusion that $c_N=0$. 
Consequently, using eq. \eqref{Vkp1} again, we conclude that $c_j=0$ and $V_j=0$ for $j\ge N$. As a 
result, $y_j$ converges to zero for $j\ge N$, and we have finite-time stability of the system. 
\end{proof}

\begin{remark}
Although the above result is given for a positive definite function $\gamma_k:= \gamma (V_k)$ satisfying 
condition \eqref{gammacond}, it holds trivially for a constant positive $\gamma$ as well. This can be easily 
verified following the first step of the proof above, by substituting $\gamma_0=\gamma=$constant in eq. 
\eqref{initLyapf}, and going through the remainder of the proof with similar arguments; $\varepsilon$ is 
not needed in this case. 
\end{remark}
The conditions given in Lemma \ref{discFTSlem} are not difficult to satisfy, as the following corollary 
shows.
\begin{corollary}\label{dFTScor}
Consider a discrete-time system with a corresponding positive definite (Lyapunov) function $V:\bR^l \to 
\bR$ and let $V_k = V (y_k)$. Let $\alpha, \varepsilon$ be constants as defined in Lemma \ref{discFTSlem}, 
and let $\gamma_k:= \gamma (V_k)$ be a class-$\cK$ function of $V_k$ that is {\bf not} class-$\cK_\infty$. 
Then, if $V_k$ satisfies the relation \eqref{discFTS}, the system is (Lyapunov) stable and $y_k$ converges 
to $y=0$ for $k\ge N$, for some $N\in\bW$.
\end{corollary}
The proof of this corollary is immediate, because if $\gamma_k$ is class-$\cK$ but {\bf not} 
class-$\cK_\infty$, then it clearly satisfies condition \eqref{gammacond} of Lemma \ref{discFTSlem} for values
of the ratio $\frac{V_k}{V_0}$ in the open interval $]0,1[$. 
%\begin{remark}\label{rem1}
%The time derivative of the (continuous-time) Lyapunov function $V(y_k)$ in the statement of Lemma 
%\ref{discFTSlem} can be approximated as:
%\be \frac{\di}{\di t}V(y_k)\approx   \frac{V_{k+1}-V_k}{t_{k+1}-t_k} = \frac{V_{k+1}-V_k}{\Delt}. 
%\label{approxVdot} \ee
%%If $\Delt \ge \tau> 0$, then
%%\[  \frac{V_{k+1}-V_k}{\Delt} \le \frac{V_{k+1}-V_k}{\tau}. \]
%Therefore the sufficient condition for finite-time stability in continuous time provided in~\cite{bhatFT}: 
%\[ \frac{\di}{\di t}V(y_k) \le -\gamma \big(V(y_k)\big)^\alpha, \;\ 0<\alpha <1, \]
%can be approximated by the sufficient condition \eqref{discFTS} in the statement of Lemma  \ref{discFTSlem}, 
%with $\gamma_k$ depending also on $\Delt$ . The inequality \eqref{gammacond} then implies that the 
%time interval between consecutive sampling instants can be lower-bounded by a positive constant $\tau$. 
%This in turn means that the sampling frequency does not need to be very high for a finite-time stable 
%observer designed for model-free control in our framework. This is unlike the model-free control framework 
%of Fliess et al~\cite{fljosi08,FlJo13}, which requires output filtering schemes operating at high sample rates, 
%as re-iterated in~\cite{griz17}. 
%\end{remark}

The following subsection gives the main result of this section in the form of a finite-time stable output 
observer in discrete time. 

\subsection{Finite-time stable output observer}\label{ssec: ftsobss}

Define the discrete-time Lyapunov function for the output observer as:
\be V^o(e^o_k) = V^o_k = \frac12 (e^o_k)\T L e^o_k, \label{discobsLyapf} \ee
where $L=L\T$ is positive definite. The total time difference of this discrete 
Lyapunov function in the time interval $[t_k, t_{k+1}]$ is then obtained as
\begin{align} 
V_{k+1}^o - V_k^o &= \frac12 (e^o_{k+1})\T L e^o_{k+1} - \frac12 (e^o_k)\T L e^o_k \nn \\
&= \frac12 \big(e^o_{k+1}- e^o_k\big)\T L \big(e^o_{k+1}+ e^o_k\big). \label{Vokdiff} 
\end{align}
An asymptotically stable observer can be designed as follows:
\begin{align} 
%\begin{split}
e^o_{k+1}= e^o_k - \beta (e^o_k + e^o_{k+1}) \mbox{ or } e^o_{k+1}= \frac{1-\beta}{1+\beta} e^o_k,
%&\mbox{where } M=M\T >0.  
%\end{split} 
\label{asympt-obs}
\end{align} 
where $\beta>0$ is a positive constant gain. The following result gives a finite-time stable output observer 
in discrete time.

\begin{theorem}\label{dFTSobs}
Let $e^o_k$ be as defined in \eqref{measmodid} and let $L$, $\beta$ be as defined in eqs. 
\eqref{discobsLyapf}-\eqref{asympt-obs}, and let $p\in ]1,2[$. Consider the discrete-time observer 
given by:
\begin{align} 
\begin{split}
&\hat y_{k+1}= y_{k+1}+ \mathcal{B}(e^o_k) e^o_k, \, \mbox{ where } \\ 
&\cB (e^o_k)= \frac{\big((e^o_k)\T Le^o_k\big)^{1-1/p} - \beta}{\big((e^o_k)\T Le^o_k\big)^{1-1/p} + \beta}. 
\end{split} \label{dFTS-obs} %,\; p\in]1,2[
\end{align}
The observer law \eqref{dFTS-obs} leads to a (Lyapunov) stable observer with convergence of the 
output estimation errors to zero for $k\ge N$ and finite $N\in\bW$.
\end{theorem}
\begin{proof}
The observer law \eqref{dFTS-obs} is equivalent to:
\be e^o_{k+1}= \mathcal{B}(e^o_k) e^o_k, \label{eqvobslaw} \ee
which gives the discrete time evolution of the output estimate error according to this observer. This can be 
re-expressed as:
\be \big((e^o_k)\T Le^o_k\big)^{1-1/p} (e^o_{k+1}-e^o_k)= -\beta (e^o_{k+1}+e^o_k). \label{eokdiff} \ee
Consider the discrete-time Lyapunov function $V^o_k$ defined by \eqref{discobsLyapf}. The difference 
between the values of this function at successive sampling instants is given by eq. \eqref{Vokdiff}. 
Substituting for $e^o_{k+1}-e^o_k$ from \eqref{eokdiff} into eq. \eqref{Vokdiff}, we get:
\be V_{k+1}^o - V_k^o = -\frac{\beta}{2} \frac{(e^o_{k+1}+ e^o_k)\T L (e^o_{k+1}+ e^o_k)}{\big((e^o_k)\T L
e^o_k\big)^{1-1/p}}. \label{diffVko} \ee
Note that $e^o_{k+1}+e^o_k=\big(1+\cB(e_k^o)\big)e^o_k$, and the right side of expression \eqref{diffVko} 
is zero if and only if
\[ e^o_{k+1}= -e^o_k, \]
which is possible if and only if $\mathcal{B}(e^o_k)=-1$, according to \eqref{eqvobslaw}. From the expression 
for $\mathcal{B}(e^o_k)$ in \eqref{dFTS-obs}, we see that $\mathcal{B}(e^o_k)=-1$ if and only if $e^o_k=0$. 
Therefore, from eqs. \eqref{dFTS-obs} and \eqref{eqvobslaw}, we see that
\[ V_{k+1}^o - V_k^o = 0 \, \Leftrightarrow\, e^o_k= 0. \]
%We can also conclude, from \eqref{eqvobslaw}, that this implies that $e^o_j=0$ for all $j\ge k$. 
 
Now substituting eq. \eqref{eqvobslaw} into the right side of eq. \eqref{diffVko} and noting that 
$(e^o_k)\T Le^o_k= 2V^o_k$, we get 
\begin{align}  
\begin{split}
V_{k+1}^o - V_k^o & = -\gamma_k \big( V_k^o\big)^{1/p}, \,\mbox{ where } \\
\gamma_k &= \frac{\beta}{2^{1-1/p}}\big(1+\cB(e_k^o)\big)^2. 
\end{split}\label{dVko2}
\end{align}
Substituting $(e^o_k)\T Le^o_k= 2V^o_k$ into the expression for $\mathcal{B}(e^o_k)$ to evaluate  
$\gamma_k$ in eq. \eqref{dVko2}, we get 
\be \gamma_k = 4\beta \frac{ 2^{1-1/p}(V^o_k)^{2-2/p}}{\big( (2V^o_k)^{1-1/p} +\beta\big)^2}. \label{gamk} \ee
%With $\alpha=1/p$, we see that eq. \eqref{gamk} implies that the inequality \eqref{gammacond} is 
%satisfied, because 
Clearly, $\gamma_k$ as given by eq. \eqref{gamk} is a class-$\cK$ function of $V^o_k$. 
From eqs. \eqref{dVko2} and \eqref{gamk}, we see that $V^o_k$ is monotonously decreasing if 
$\gamma_k>0$ and 
%$\gamma_k \to 4\beta/2^{1-1/p}$ from below as $V^o_k\to \infty$.
\[ 0< \gamma_k < \frac{4\beta}{2^{1-1/p}}\ \mbox{ for }\ 0 < 2V_k^o <\infty. \]
Therefore $\gamma_k$ is class-$\cK$ but not class-$\cK_\infty$, and therefore satisfies the stronger condition 
of Corollary \ref{dFTScor}. To explicitly show this, from eq. \eqref{gamk} we obtain the ratio: 
%\begin{align}
%\begin{split} 
\be 
%\frac{\gamma_k}{(V^o_k)^{1-1/p}} = 4\beta \frac{ (2V^o_k)^{1-1/p}}{\big( (2V^o_k)^{1-1/p} +\beta\big)^2}. 
a_k:=\frac{\gamma_k}{\gamma_0}= \frac{(V^o_k)^{2-2/p}}{(V^o_0)^{2-2/p}}\frac{\big( (2V^o_0)^{1-1/p} +
\beta\big)^2}{\big( (2V^o_k)^{1-1/p} +\beta\big)^2} <1.
\label{gamVk} \ee
%\mbox{because }\, & \Leftrightarrow 
%\end{split}\label{gamVk} 
%\end{align}
For values of the ratio $\frac{V_k^o}{V_0^o}$ in the open interval $]\chi,1[$ where $0<\chi\ll 1$, the ratio in eq. 
\eqref{gamVk} is bounded below according to:
\begin{align}
\begin{split}
a_k & > \frac{\chi^{1-1/p} (1+\mu)}{\chi^{1-1/p}+ \mu}  \\
& = \Big\{ 1- \frac{\mu (1- \chi^{1-1/p})}{\chi^{1-1/p}+ \mu} \Big\}^2
\end{split} \mbox{ where } \mu= \frac{\beta}{(2V_0^o)^{1-1/p}}. \label{aklowbnd} 
\end{align}
This guarantees the existence of $\varepsilon\in\, ]0,1[$ that satisfies the condition \eqref{gammacond} 
in the statement of Lemma \ref{discFTSlem} for $V^o_k\in ]\chi V_0^o, V^o_0[$, and is given by:
\[ \varepsilon= 2\delta -\delta^2,  \mbox{ where } 0<\delta= \frac{\mu (1- \chi^{1-1/p})}{\chi^{1-1/p}+ \mu} <1. \] 
Therefore, $V^o_k=0$ for $k\ge N$ for some finite $N\in\bW$, so this discrete-time nonlinear observer 
ensures finite-time stable convergence of output estimation errors to zero. 
\end{proof}

%For this finite-time stable observer, we can express the observer law as the following:
%\be
%\hat x_{k+1} = \eta_{k+1} (Y_{k+1}) + \cB (e^o_k) \big( \hat x_k -\eta_k (Y_k)\big), \label{FTSobslaw} 
%\ee
%where $\cB (e^o_k)$ is expressed in eq. \eqref{FTS-obs}. This discrete-time observer ensures that the 
%estimated states converge to actual (true) states given by $x_k= \eta (Y_k)$ in finite time. 
As remarked before, finite-time stability is advantageous compared to asymptotic stability for added  
robustness to disturbances and noise in the measurements $y_k$. Moreover, finite-time stability of 
the observer facilitates a separation between observer design and controller design. This observer is 
also robust to initial estimate errors (i.e., magnitude of $e^o_0$) if no initial measurements are available 
or if there is poor knowledge of initial output.

\section{The Ultra-Local Model and Its Estimation}\label{sec: ulm}
\subsection{Ultra-Local Model for Unknown Input-Output Dynamics}\label{ssec: ulmd}
The model-free control approach of~\cite{FlJo13} relates the unknown model of the dynamics to an 
order $\nu$ ``ultra-local model." Here, we generalize the ultra-local model to the form:
\be y_k^{(\nu)} = \cF_k +\cG_k  u_k,\, \mbox{ where }\, \cF_k\in\bR^l,\; u_k\in\bR^m, \label{ultramodel} \ee
and $\cG_k \in \bR^{l\times m}$ is a full rank matrix that is selected appropriately, as part of the 
controller design. However, the approach of~\cite{FlJo13} deals with SISO systems using techniques 
from classical control and does not consider stability or robustness of the feedback control system. In 
contrast, the approach given here is centered around provable guarantees on nonlinear stability and 
robustness to external disturbances and measurement noise. To do this in an effective manner, the 
unknown input-output dynamics, captured by $\cF_k\in\bR^l$ in eq. \eqref{ultramodel}, needs to be 
estimated in a stable and robust manner. We therefore consider the following problem. 
 
%For second-order systems under consideration, $\nu=2$, 
%and $\ddot y_k$ is approximated be a second-order finite difference as in eq. \eqref{dynoutmod}. 
%The vector $\cF_k$ is an unknown that describes the part of the local output behavior in time that 
%does not directly depend on the input. 
\begin{problem}\label{prob1}
Consider the unknown nonlinear system \eqref{dynoutmod} satisfying Assumptions \ref{assump1} and 
\ref{assump2}, with control inputs $u_k:=u(t_k)\in\bR^m$ provided at 
discrete sample times $t_k$. Given the discrete time ultra-local model \eqref{ultramodel} of the input-output  
dynamics with unknown $\cF_k$, estimate $\cF_k$ from past input-output history and design a feedback 
control scheme to track the desired output trajectory $y^d_k:= y^d(t_k)$ in a nonlinearly stable manner.  
\end{problem} %where $y^d(t)$ is continuous and $\nu$ times differentiable, i.e., $y^d(t)\in\mC^\nu (\bR^l)$

Note that as per Assumption \ref{assump1}, the system is input-output controllable. %In the case that the 
%system is a ``gray box" where the controls parameterization and influence on the dynamics (i.e., $\cG_k$) 
%is known but the uncontrolled dynamics (i.e., $\cF_k$) is not known, then 
In the following subsections of this section, we design two nonlinear observers to estimate $\cF_k$ for later 
use in the output feedback tracking control scheme. These schemes (in isolation) can also be used to identify 
this unknown dynamics using known (feedforward) control inputs $u_k$ and influence matrix $\cG_k$. Such 
a situation can be useful in applications where the control parameterization is well-known, but the dynamics 
is influenced by external disturbances or internal parameters that are unknown. Note that the model given 
by \eqref{ultramodel} is a generalization of the ultra-local model of~\cite{FlJo13}, where $\cG_k$ was a 
constant scalar and only single-input single-output (SISO) systems were considered. 
%We consider the estimation of $\cF_k$ in subsection \ref{ssec: estimdyn} and tracking control to control 
%a desired trajectory $y^d_k$ in subsection \ref{ssec: trajcon}.

%\begin{remark}\label{remulm}
%Note that the ultra-local model \eqref{ultramodel} models the unknown (possibly nonlinear) system with 
%a linear model, so the usefulness of nonlinear control based on this model is not immediately clear. The 
%primary reason we use nonlinear feedback is due to our controller and observer designs, which provide 
%finite-time convergence and therefore add robustness to measurement noise and disturbances. Another 
%reason for use of nonlinear model-free control would be if the space of output variables is not a vector 
%space, in which case it may be a manifold with or without boundary. In such a situation, the difference 
%equation of the ultra-local model \eqref{ultramodel} would be modified to ensure that the geometry of 
%the output space is maintained.
%\end{remark}

\subsection{Estimation of Unknown Input-Output Dynamics Using a First Order Observer}\label{ssec: estF1st}
The model-free intelligent PID (iPID) control framework of~\cite{FlJo13} does not provide a nonlinearly 
stable observer scheme to estimate the unknown input-output dynamics that is not directly influenced 
by the control inputs. Here, we provide a first-order observer for this unknown dynamics, i.e.,  
$\cF_k$ in eq. \eqref{ultramodel}. The idea here is to use the finite-time stable output observer design 
outlined in the previous section in conjunction with a first-order hold to estimate the unknown dynamics 
expressed by $\cF_k$ in eq. \eqref{ultramodel} based on past input-output history. Note that the control 
law for $u_{k+1}$ cannot be based on knowledge of $\cF_{k+1}$ which is unknown due to causality; 
but it can use past information on $\cF_j$ for $j\in\{0,\ldots,k\}$. The control law uses the predicted value 
of $\cF_k$, which we denote $\hat\cF_k$, to construct the control $u_k$. %This process is described next.

Define the estimation error in estimating $\cF_k$ as follows:
\be e^\cF_k := \hat\cF_k- \cF_k. \label{esterrF} \ee
The following result gives a first order (discrete time) nonlinearly stable observer for the unknown 
dynamics $\cF_k$. 
\begin{proposition}\label{prop1st}
Let $e^\cF_k$ be as defined by eq. \eqref{esterrF}, and let $r \in ]1,2[$ and $\lambda >0$ be constants. 
Let the first order finite difference of the unknown quantity $\cF_k$, given by 
\be \Delta\cF_k:= \cF_k^{(1)}= \cF_{k+1}- \cF_k, \label{DelF} \ee
be bounded. Consider the nonlinear observer given by:
\begin{align}
\begin{split}
&\hat\cF_{k+1} = \cD (e^\cF_k) e^\cF_k + \cF_k, \\
&\mbox{where }  \cD (e^\cF_k )= \frac{\big((e^\cF_k)\T e^\cF_k\big)^{1-1/r} -\lambda}{\big((e^\cF_k)\T 
e^\cF_k\big)^{1-1/r} +\lambda}.
\end{split} \label{Fest1st}
\end{align}
This observer leads to finite time stable convergence of the estimation error vector $e^\cF_k \in\bR^l$ to a 
bounded neighborhood of zero, where the bounds are given by the bounds on $\Delta\cF_k$.  
\end{proposition}
\begin{proof}
The proof of this result begins by showing that if 
\be e^\cF_{k+1}= \cD (e^\cF_k) e^\cF_k, \label{FestFTS} \ee
where $\cD (e^\cF_k)$ is as defined by eq. \eqref{Fest1st}, then $e^\cF_k$ converges to zero in a 
finite-time stable (FTS) manner. This can be shown by defining the discrete-time Lyapunov function
\[ V^\cF_k := (e^\cF_k)\T e^\cF_k. \]
Taking the discrete time difference of this Lyapunov function, we get
\begin{align} 
\begin{split}
&V^\cF_{k+1}- V^\cF_k = -\gamma^\cF_k (V^\cF_k)^{1/r} \\
&\mbox{where } \gamma^\cF_k= \lambda\big( 1+\cD (e^\cF_k)\big)^2.
\end{split} \label{diffVcF} 
\end{align}
It can be easily verified (in a manner similar to that for $\gamma_k$ in the proof of Theorem \ref{dFTSobs}) 
that this $\gamma^\cF_k$ satisfies the sufficient condition of Corollary \ref{dFTScor} for finite-time stability 
of $e^\cF_k$. Using the definition of $e^\cF_k$ given by eq. \eqref{esterrF} and the relation \eqref{FestFTS}, 
one obtains the following observer for $\hat\cF_k$:
\be \hat\cF_{k+1}= \cD (e^\cF_k) e^\cF_k + \cF_{k+1}. \label{idealFest} \ee
However, as mentioned earlier, $F_{k+1}$ is not available at time $t_{k+1}$ due to causality; therefore, it 
needs to be replaced by a known quantity. This first order observer design given by eq. \eqref{Fest1st} 
replaces $F_{k+1}$ in eq. \eqref{idealFest} with $F_k$. As a result, the estimation error $e^\cF_k$ evolves 
according to:
\begin{align}
\begin{split} 
 &e^\cF_{k+1} := \hat\cF_{k+1}- \cF_{k+1} = \cD (e^\cF_k) e^\cF_k + \cF_k- \cF_{k+1} \\
 &= \cD (e^\cF_k) e^\cF_k - \Delta\cF_k,\, \mbox{ where } \Delta\cF_k=\cF_{k+1}-\cF_k. 
\end{split}\label{err1st}
\end{align}
Therefore this observer is a first order perturbation of the ideal FTS observer design for $\cF_k$ as 
given by eq. \eqref{idealFest}, with the perturbation coming from the first oder difference term 
$\Delta\cF_k$. Due to the FTS behavior of this ideal observer for $\cF_k$, the first order observer 
design of eq. \eqref{Fest1st} will converge to a neighborhood of $e^\cF_k=0$, where the size of this 
neighborhood is given by the bounds on $\Delta\cF_k$. For example, if $\|\Delta\cF_k\|$ is bounded 
by a known constant, then $e^\cF_k$ will remain bounded by this constant after a finite time interval. 
Clearly, the smaller the bounds on $\Delta\cF_k$, the smaller the neighborhood of $e^\cF_k=0$ that 
this observer will converge to within finite time. 
\end{proof}

\begin{remark}
This first oder observer can become unstable if $\Delta\cF_k$ escapes (becomes unbounded) in 
finite time at a rate faster than that dictated by the design of $\cD (e^\cF_k)$. However, the design  
of a model-free control scheme for such a system is beyond the scope of this work.
\end{remark}

\begin{remark}
For use in conjunction with the FTS output observer given by Theorem \ref{dFTSobs} for feedback control, 
the gain parameters $\alpha$ and $\lambda$ in $\cD(e^\cF_k)$ can be designed such that the convergence 
rate of the ideal FTS observer for $\cF_k$ (given by \eqref{idealFest}) is slower than that of the FTS output 
observer.
\end{remark}

\subsection{Estimation of Unknown Input-Output Dynamics Using a Second Order 
Observer}\label{ssec: estF2nd}
In this subsection, we design a second order observer for $\cF_k$ based on the developments in the 
previous subsection. To start the design process, we assume an internal dynamics model for $\cF_k$ 
given by:
\be \cF_{k+1}=\cF_k + \Delta\cF_k, \label{intmodcF} \ee
where $\Delta\cF_k$ is as defined in eq. \eqref{DelF}. The second order observer design is based on 
the above model, as follows:
\be \hat\cF_{k+1}= \hat\cF_k +\Delta\hat\cF_k, \label{est2nd} \ee 
where $\Delta\hat\cF_k$ is the estimate of $\Delta\cF_k$. In addition, define the error in estimating 
$\Delta\cF_k$ as follows:
\be e^\Delta_k := \Delta\hat\cF_k - \Delta\cF_k. \label{erDelcF} \ee
The following result gives the second order observer for $\cF_k$ based on a particular selection for 
$\Delta\hat\cF_k$. 

\begin{proposition}\label{prop2nd}
Let $e^\Delta_k$ be as defined by eq. \eqref{erDelcF}, and $e^\cF_k$, $\alpha\in ]0,1[$ and $\lambda >0$ 
be as defined in Proposition \ref{prop1st}. Further, let $\cD (\cdot)$ be as defined by eq. \eqref{Fest1st} 
in Proposition \ref{prop1st}, and let the second order finite-time difference given by: 
\be \Delta^2\cF_k:= \cF_{k-1}^{(2)}= \cF_{k+1}- 2\cF_k +\cF_{k-1} \label{Del2F} \ee
be bounded. Consider the nonlinear observer given by:
\begin{align}
\begin{split}
&\hat\cF_{k+1} = \cD (e^\cF_k) e^\cF_k + F_k + \Delta\hat\cF_k, \\
&\mbox{where }  \Delta\hat\cF_k= \cD (e^\Delta_{k-1}) e^\Delta_{k-1} + \Delta\cF_{k-1}.
\end{split} \label{Fest2nd}
\end{align}
This observer leads to finite time stable convergence of the estimation errors $e^\cF_k, e^\Delta_k \in
\bR^l$ to bounded neighborhoods of zero, where the bounds are given by bounds on  $\Delta^2\cF_k$. 
\end{proposition}
\begin{proof}
The proof of this result starts by noting that the ideal FTS observer law for $\cF_k$ given by eq. 
\eqref{idealFest} can also be expressed as:
\be \hat\cF_{k+1}= \cD (e^\cF_k) e^\cF_k + \cF_k+ \Delta\cF_k, \label{idealFest2} \ee
because the last two terms on the right side of this expression add up to $\cF_{k+1}$. The second order 
observer law given by eq. \eqref{Fest2nd} is obtained by replacing $\Delta\cF_k$ on the RHS of eq. 
\eqref{idealFest2} with its estimate. The estimate $\Delta\hat\cF_k$ will converge to the true value 
$\Delta\cF_k$ in finite time, if it was updated according to the (ideal) observer law:
\be \Delta\hat\cF_k= \cD (e^\Delta_{k-1}) e^\Delta_{k-1} + \Delta\cF_k. \label{idealDFest} \ee
Note that this ideal observer for $\Delta\cF_k$ is of the same form as the ideal FTS observer law for 
$\cF_k$ given by eq. \eqref{idealFest}. And like the ideal observer \eqref{idealFest}, the observer eq. 
\eqref{idealDFest} is not practically implementable because $\Delta\cF_k$ is unknown at time $t_k$ 
(because $\cF_k$ is unknown). As we did with the first order observer in Proposition \ref{prop1st}, we 
replace $\Delta\cF_k$ in \eqref{idealDFest} with its previous value, assuming that this quantity changed 
little in the time interval $[t_{k-1},t_k]$. This leads to the following observer law for $\Delta\cF_k$: 
\be \Delta\hat\cF_k= \cD (e^\Delta_{k-1}) e^\Delta_{k-1} + \Delta\cF_{k-1}. \label{DFest} \ee
The resulting second order observer is therefore given by eqs. \eqref{Fest2nd}. To show that this is 
indeed second order, the evolution of the estimation error $e^\cF_k$ in discrete time is obtained as 
below:
\begin{align}
\begin{split} 
 e^\cF_{k+1} &:= \hat\cF_{k+1}- \cF_{k+1} = \cD (e^\cF_k) e^\cF_k  +\cD (e^\Delta_{k-1}) e^\Delta_{k-1} \\ 
 &\; \; + \cF_k + \Delta\cF_{k-1} - \cF_{k+1} \\
 &= \cD (e^\cF_k) e^\cF_k +\cD (e^\Delta_{k-1}) e^\Delta_{k-1} - \Delta^2\cF_k,
\end{split}\label{err2nd}
\end{align}
where $\Delta^2\cF_k$ is as defined by eq. \eqref{Del2F}. The last line in the above expression is 
obtained by substituting for $\Delta\cF_{k-1}$ in the previous line, using the definition of $\Delta\cF_k$ 
given by eq. \eqref{DelF}. The remainder of the proof of this result uses the same arguments as in the 
last part of the proof of Proposition \ref{prop1st}, with $\Delta\cF_k$ replaced by $\Delta^2\cF_k$. 
\end{proof}

\begin{remark}
It is clear from the constructive proofs of Propositions \ref{prop1st} and \ref{prop2nd} that higher order 
observers for $\cF_k$ may be constructed using a similar process as outlined in these proofs. For 
example, a third order observer can be constructed by replacing $\Delta\cF_{k-1}$ in the second line 
of eq. \eqref{Fest2nd} with $\Delta\cF_{k-1}+ \Delta^2\hat\cF_{k-1}$ and finding an appropriate update 
law for $\Delta^2\hat\cF_{k-1}$. Clearly, the added computational burden of higher order observers 
make them unattractive for implementation when the higher order differences of the discrete signal 
$\cF_k$ are known to be within reasonable bounds. In most situations when bounds (perhaps 
conservative) on $\Delta\cF_k$ and $\Delta^2\cF_k$ are known, these low order observers are 
adequate.
\end{remark}

\section{Model-free nonlinearly stable feedback tracking control}\label{sec: mfccon} 
This section develops a nonlinear model-free output feedback tracking control scheme that solves 
Problem \ref{prob1} in Section \ref{ssec: ulmd}. The control design process is based on Assumptions 
\ref{assump1} and \ref{assump2} for the nonlinear system \eqref{dynoutmod} (or \eqref{discmod} in 
discrete time), and is designed to track a desired output trajectory for a system expressed by the 
ultra-local model \eqref{ultramodel}. The control design 
given here may make use of the finite-time stable output observer developed in Section \ref{ssec: ftsobss} 
to filter out noise in output measurements, as well as the nonlinear observers for the ultra-local model 
given in Sections \ref{ssec: estF1st} and \ref{ssec: estF2nd}. But the control design is independent of 
these observers designed in the earlier sections, and can be used in conjunction with other output 
and ultra-local model observers that do the same tasks. This control scheme, however, does need 
accurate estimates of the ultra-local model and the measured outputs for output tracking.

%\subsection{The Ultra-Local Model}\label{sec: ulm}
\subsection{Output Trajectory Tracking Control}\label{ssec: trajcon}
%and $x^d (t)\in \mC^k (\bR^n)$ is a desired continuous state trajectory that is $k$ times differentiable 
%with $k\ge 2$. 
Our framework for nonlinear model-free control %is depicted as a block diagram in Figure \ref{BDfig}.
designs a control law for the control input $u_k$ at time $t_k$ from the output estimate $\hat y_k$, the 
desired output $y^d_k$, and the estimate of the ultra-local model $\hat\cF_k$ constructed from output 
measurements with additive sensor noise and past input-output history as described in Sections 
\ref{sec: ftsobs} and \ref{sec: ulm}. 
%\begin{figure}[htb]
%\begin{center}
%\hspace*{-2mm}\includegraphics[width=\columnwidth]{nMFC-BDiag.pdf}
%\caption{Block diagram representing our nonlinear model-free control framework.}  % width is 8.4 cm.
%\label{BDfig}                                 % Size the figures 
%\end{center}        
%\vspace*{-2mm}                         % accordingly.
%\end{figure}
Considering Problem \ref{prob1}, define the output trajectory tracking error 
\be e_k= y_k - y^d_k\, \mbox{ where }\, y^d_k= y^d (t_k). \label{trackerrs} \ee
In practice, the true output $y_k$ is substituted by its estimate $\hat y_k$ for feedback tracking control 
of $y^d_k$. The objectives of the control design are: (1) to ensure that the feedback system tracks the 
desired trajectory in a nonlinearly stable manner; and (2) in the absence of measurement noise and if 
$\hat\cF_k =\cF_k$, $e_k$ converges to the zero vector asymptotically or in finite time starting from a 
finite non-zero value. 
%; therefore, the error variable
%\be \hat e^c_k=  \hat y_k - y^d_k \label{trackerest} \ee
%is used for feedback.

For an unknown system whose input-output behavior is modeled by the ultra-local model \eqref{ultramodel}, 
we define the variable:
\begin{align} 
\begin{split}
&s_k = e_k^{(\nu -1)} + c_1 e_k^{(\nu -2)}+ \cdots + c_{\nu-1} e_k,\\ 
&\mbox{where } 1>c_1>\ldots>c_{\nu-1}> 0. %\mbox{ for } i=1,\ldots,\nu-1. 
\end{split} \label{skdefn} 
\end{align}
This variable plays a role similar to a sliding mode in sliding mode control. As in \eqref{ultramodel}, 
$(\cdot)^{(\mu)}$ represents the discrete-time analog of the $\mu$th time derivative, so that $s_k=0$ is  
a $(\nu-1)$ order finite difference equation. Note that the condition on the $c_i$ for $i=1,\ldots,\nu-1$ 
ensures that $z^{\nu-1}+ c_1 z^{\nu-2}+ \cdots+ c_{\nu-1}$ is a Schur (stable) polynomial, and therefore 
the manifold $s_k=0$ has $e_k=0$ as the globally exponentially stable equilibrium. Thereafter, the control 
design process ensures that the feedback system converges in a finite-time stable manner to the manifold 
$s_k=0$. 

\begin{remark}\label{skrem}
Although the process outlined in the previous paragraph is similar to that followed in sliding mode 
control approaches, there is a key difference. The approach followed here, as with our observer design, 
is to obtain finite-time stable convergence to a desired equilibrium or manifold in a manner that is 
continuous (in this case H\"{o}lder continuous). This approach avoids the disadvantages of discontinuous 
feedback control like chattering, non-standard notions of solutions, and implementation issues with 
actuators that can only provide continuous control inputs. 
\end{remark}  

With the variable $s_k$ as defined by \eqref{skdefn}, the control law design proceeds by defining the 
Lyapunov function 
\be V^c_k = \frac12 (s_k)\T K s_k, \label{disconLyapf} \ee
where $K=K\T\in\bR^{n\times n}$ is a positive definite matrix, which makes $V^c_k$ a positive definite 
function of $s_k$. The total time difference of this discrete Lyapunov function in the time interval 
$[t_k, t_{k+1}]$ is then obtained as
\begin{align} 
V_{k+1}^c - V_k^c &= \frac12 s_{k+1}\T K s_{k+1} - \frac12 s_k\T K s_k \nn \\
&= \frac12 \big(s_{k+1}- s_k\big)\T K \big(s_{k+1}+ s_k\big). \label{Vckdiff} 
\end{align}
A sufficient condition for $s_k$ to converge to zero in an asymptotically stable manner is to ensure that
\be  
s_{k+1}= s_k - \eta (s_k + s_{k+1}) \mbox{ or } s_{k+1}= \frac{1-\eta}{1+\eta} s_k,
\label{asympt-con}
\ee 
where $\eta>0$ is a constant positive control gain. Due to the definition of $s_k$ given by 
\eqref{skdefn}, this in turn ensures that the feedback system is exponentially convergent in discrete-time, 
i.e., $e_k$ converges to zero exponentially so that the desired output trajectory is tracked exponentially. 
The following statement gives a finite-time stable control law in discrete time. 

\begin{lemma}\label{dFTS-sk}
Let $s_k$ be as defined in \eqref{skdefn} and let $\eta>0$ and $q\in]1,2[$. Let the discrete-time 
evolution of $s_k$ be given by:
\begin{align} 
s_{k+1}= \mathcal{C}(s_k) s_k, \, \mbox{ where }\,  \cC (s_k)= \frac{\big((s_k)\T s_k
\big)^{1-1/q} - \eta}{\big((s_k)\T s_k\big)^{1-1/q} + \eta}. \label{dFTS-cont} %,\; p\in]1,2[
\end{align}
The discrete-time evolution \eqref{dFTS-cont} leads to (Lyapunov) stable convergence of $s_k$ to 
zero for $k\ge N$ and finite $N\in \bW$.
\end{lemma}
\begin{proof}
Consider the difference of the Lyapunov function given by \eqref{Vckdiff}. Substituting eq. \eqref{dFTS-cont} 
into this expression, we get:
\be V_{k+1}^c - V_k^c = -\frac{\eta}{2} \frac{(s_{k+1}+ s_k)\T (s_{k+1}+ s_k)}{\big((s_k)\T s_k\big)^{1-1/p}}. \label{diffVkc} \ee 
The remainder of this proof is identical to the proof of Theorem \ref{dFTSobs}, with $e^o_k$ replaced 
by $s_k$; $\beta$, $p$, and $L$ replaced by $\eta$, $q$, and $I$ (the identity matrix) respectively; and 
corresponding changes. 
\end{proof}

The following statement gives a control law that 
ensures that $s_k$ converges to zero in finite time, and therefore the output tracking error $e_k$ converges 
to zero exponentially, if $\hat\cF_k$ converges to $\cF_k$ in finite time. 

\begin{theorem}\label{dFTScont}
Consider $s_k$ as defined by eq. \eqref{skdefn} and define $e_k^{(\nu)}$ recursively from $e_k^{(\nu-1)}$ 
as in eq. \eqref{fordif}. %follows:
%\be e_k^{(\nu)} := e_{k+1}^{(\nu-1)}-e_k^{(\nu-1)} \mbox{ with } e_k^{(0)}=e_k. \label{eknudef} \ee
Thereafter, consider the control law:
\begin{align}
\begin{split} 
\cG_k  u_k =& \big(y^d_k\big)^{(\nu)} - \frac{2\eta}{\big((s_k)\T s_k\big)^{1-1/q} + \eta}s_k- \hat\cF_k \\
%\frac{s_k}{\Delt}
&- c_1 e_k^{(\nu-1)} -\ldots -c_{\nu -1} e_k^{(1)},
\end{split} \label{MFClaw} 
\end{align} %-\hat\cF_k + \cB(e^{\cF}_{k-1}) e^{\cF}_{k-1}
where $\eta$ and $q$ are as defined in Lemma \ref{dFTS-sk}. 
%$\dot e_k= e_k^{(1)}$ is the discrete-time approximation to the first time derivative of $e_k$. 
Then the unknown system with the ultra-local model \eqref{ultramodel} and the control law \eqref{MFClaw} 
tracks the desired output trajectory $y^d_k$ in an exponentially stable manner if $\hat\cF_k$ converges to 
$\cF_k$ in finite time. 
\end{theorem}  
\begin{proof}
To start with, we re-express eq. \eqref{dFTS-cont} as follows:
\be s_{k+1}-s_k = -\frac{2\eta}{\big((s_k)\T s_k\big)^{1-1/q} + \eta} s_k. \label{reskeq} \ee
%Dividing both sides of \eqref{reskeq} by the time step-size $\Delt$ and 
Substituting eq. \eqref{skdefn} and using the recursive definition of $e_k^{(\nu)}$, we see that  %\eqref{eknudef},
\be e^{(\nu)}_k + c_1 e^{(\nu-1)}_k+ \ldots+ c_{\nu-1} e_k^{(1)}= -\frac{2\eta}{\big((s_k)\T s_k\big)^{1-1/q} + 
\eta} s_k. %\frac{s_k}{\Delt}.
\label{mfcslidm} \ee
Now substituting the ultra-local model \eqref{ultramodel} and eq. \eqref{trackerrs}, we get:
\begin{align}
\begin{split} 
&\cF_k+ \cG_k u_k  -(y^d_k)^{(\nu)} + c_1 e^{(\nu-1)}_k+ \ldots+ c_{\nu-1} e_k^{(1)} \\ 
&= -\frac{2\eta}{\big((s_k)\T s_k\big)^{1-1/q} + \eta} s_k. %\frac{s_k}{\Delt}.
\end{split} \label{nmfcsk} 
\end{align}
Replacing $\cF_k$ with $\hat\cF_k$ in eq. \eqref{nmfcsk} and re-arranging terms, we obtain the control law
\eqref{MFClaw} for the system. Therefore, if $\hat\cF_k$ converges to $\cF_k$ in finite time (i.e., for finite 
$k$), then $s_k$ converges to zero in finite time and thereafter $e_k$ converges to zero exponentially. 
\end{proof}

\subsection{Robustness of Model-Free Output Tracking Control Scheme}
The convergence of the output tracking error $e_k$ to zero is contingent upon $\hat\cF_k$ converging to 
$\cF_k$ in finite time, according to Theorem \ref{dFTScont}. This remains true even if $e_k$ is replaced 
by $\hat e_k:= \hat y_k- y^d_k$ in the definition \eqref{skdefn} where $\hat y_k$ is obtained by the 
finite-time stable output observer of Theorem \ref{dFTSobs}. This is because if $\eta\in ]0,\beta[$ and $q
\in]1,p[$, the output estimation error $e^o_k$ given by the output observer eq. \eqref{dFTS-obs} converges 
to zero in a shorter time horizon than $s_k$ converges to zero according to eq. \eqref{dFTS-sk}. 
However, in practice, observers for the ultra-local model $\cF_k$ like those in Propositions \ref{prop1st} and 
\ref{prop2nd} can at best ensure stable convergence of $\hat\cF_k$ to a neighborhood of $\cF_k$ in finite 
time, with the size of this neighborhood depending on the size of the first-order difference $\Delta\cF_k$ or 
second-order difference $\Delta^2\cF_k$, respectively. Therefore, when the controller of Theorem 
\ref{dFTScont} is used with $\hat\cF_k$ given by the ultra-local model observers of Proposition \ref{prop1st} 
or \ref{prop2nd}, the overall output tracking scheme will be Lyapunov stable, but not exponentially 
(or asymptotically) stable. This is shown in the following corollary to Theorem \ref{dFTScont}. 

\begin{corollary}\label{robcor}
The feedback tracking control law given by eq. \eqref{MFClaw} used in conjunction with either of the 
ultra-local model observers given by eqs. \eqref{Fest1st} or \eqref{Fest2nd}, lead to the feedback 
system being (Lyapunov) stable and robust to errors in the ultra-local model estimate, 
$e^\cF_k$. 
\end{corollary}
\begin{proof}
This is shown by substituting the feedback control law \eqref{MFClaw} into the ultra-local model for the 
input-output dynamics given by eq. \eqref{ultramodel}. That leads to the expression:
\begin{align}
y_k^{(\nu)}=& \cF_k + (y^d_k)^{(\nu)} -c_1 e_k^{(\nu -1)} -\ldots -c_{\nu-1} e_k^{(1)}- \hat\cF_k \nn \\
& -\frac{2\eta}{\big((s_k)\T s_k\big)^{1-1/q} + \eta} s_k. \label{yknu} 
\end{align}
Re-arranging eq. \eqref{yknu} to express in terms of $e^\cF_k$, $e_k$ and its finite differences, we get:
\begin{align}
& e^\cF_k + e_k^{(\nu)}+ c_1 e_k^{(\nu -1)} +\ldots +c_{\nu-1} e_k^{(1)} \nn \\
& +\frac{2\eta}{\big((s_k)\T s_k\big)^{1-1/q} + \eta} s_k =0 \nn \\
& \Rightarrow e^\cF_k + s_k^{(1)} +\frac{2\eta}{\big((s_k)\T s_k\big)^{1-1/q} + \eta} s_k =0. \label{errFB}
\end{align}
Noting that $s_k^{(1)}= s_{k+1}-s_k$ according to the finite difference defined by eq. \eqref{fordif}, we see 
that the expression \eqref{errFB} is a perturbation of the ideal finite-time stable behavior of $s_k$ as 
given by eq. \eqref{reskeq}, where the perturbing signal is $e^\cF_k$. Therefore $s_k$ converges to a 
neighborhood of zero, where the size of this neighborhood depends on the size of $e^\cF_k$. Now 
invoking Proposition \ref{prop1st} or Proposition \ref{prop2nd}, we see that $e^\cF_k$ remains ultimately 
bounded if the first or second order differences $\Delta\cF_k$ or $\Delta^2\cF_k$ are bounded, respectively. 
As these finite differences will be bounded according to Assumption \ref{assump2} and the ultra-local 
model \eqref{ultramodel}, this concludes the proof. 
\end{proof}

\begin{remark}
Note that Theorem \ref{dFTScont} or Corollary \ref{robcor} do not specify how to select control gains with 
respect to the previous results on output and ultra-local model observers given in sections \ref{sec: ftsobs} 
and \ref{sec: ulm}. To ensure stability of the overall loop, it is necessary to ensure that the output observer 
converges the fastest, so that $\hat y_k$ converges to $y_k$ faster than $\hat\cF_k$ converges to (a 
neighborhood of) $\cF_k$ or $s_k$ converges to zero. Further, it is useful to ensure that the function 
$\cC(s_k)$ in the control design gives slower convergence of $s_k$ towards zero than the function 
$\cD(e^\cF_k)$ in the ultra-local model observer designs of Section \ref{sec: ulm}. This will lead to 
$\hat\cF_k$ converging to a desired neighborhood of $\cF_k$ faster than $s_k$ approaches zero (e.g., 
when conservative bounds on $\Delta\cF_k$ and $\Delta^2\cF_k$ are known, as mentioned in Remark 3).
\end{remark}

\subsection{Output Trajectory Tracking for Second Order System}
The final result given here is a model-free control law for a general second-order input-output system. 
Assuming that $\nu=2$ in the ultra-local model \eqref{ultramodel}, we define $s_k$ as follows:
\be s_k= e_{k+1}- e_k +\mu e_k\, \mbox{ where }\, 0<\mu<1. \label{sk2nd} \ee
The discrete-time ultra-local model for a second order system ($\nu=2$) is obtained from  
%approximating the second order derivative in 
eq. \ref{ultramodel} as follows:
\be y_k^{(2)}=  y_{k+2}- 2y_{k+1} + y_k= \cF_k + \cG_k u_k. 
\label{y2dot} \ee %\ddot y_k= \frac{y_{k+1}-2 y_k+ y_{k-1}}{\Delt^2}.
With this ultra-local model and $s_k$ as defined by \eqref{sk2nd}, we have the following result.
\begin{corollary}\label{cont2nd}
Consider the second-order discrete-time system given by \eqref{y2dot} with $s_k$ as defined by 
\eqref{sk2nd} and with $e_k$ defined by \eqref{trackerrs}. %define
%\be e_k^{(1)} := \dot e_k= \frac{e_k- e_{k-1}}{\Delt} \mbox{ where } e_k= y_k- y_k^d. \label{ek1def} \ee 
Then this system with the control law:
\begin{align} 
\cG_k u_k =& y^d_{k+2}-2 y^d_{k+1}+ y^d_k - \frac{2\eta}{\big((s_k)\T s_k\big)^{1-1/q} + \eta}e_k^{(1)} \nn \\
&+ \cC (s_k) \mu e_k -\mu e_{k+1} -\hat\cF_k, \label{MFClaw2nd} 
%\frac{\big((s_k)\T s_k\big)^{1-1/p} -\alpha}{\big((s_k)\T s_k\big)^{1-1/p} 
%+ \alpha}\frac{\mu e_{k-1}}{\Delt} - \frac{\mu e_k}{\Delt} -\cF_k^-, \label{MFClaw2nd} 
\end{align}
tracks the desired output trajectory $y^d_k$ in an exponentially stable manner if $\hat\cF_k$ converges 
to $\cF_k$ in finite time. 
\end{corollary}
\begin{proof}
The proof is based on showing equivalence of the control law \eqref{MFClaw2nd} for this second-order 
system with the more general expression \eqref{MFClaw} in Theorem \ref{dFTScont}, given eqs. 
\eqref{sk2nd}-\eqref{y2dot}. Substituting $\nu=2$ %eqs. \eqref{y2dot} and \eqref{ek1def} 
into the right hand side of eq. \eqref{MFClaw} and noting that $c_1=\mu$, we obtain: 
\begin{align} 
\cG_k u_k=& y^d_{k+2}-2 y^d_{k+1}+ y^d_k - \frac{2\eta}{\big((s_k)\T s_k\big)^{1-1/q} + \eta}s_k \nn \\
&- \mu (e_{k+1}- e_k) -\hat\cF_k. \label{intMFC2nd} 
\end{align}
Now substituting for $s_k$ from eq. \eqref{sk2nd} into the numerator of the fractional term on the right 
hand side of expression \eqref{intMFC2nd}, we obtain the control law \eqref{MFClaw2nd} for this second 
order system. Therefore, according to Theorem \ref{dFTScont}, the feedback system given by eqs. 
\eqref{y2dot} and the control law \eqref{MFClaw2nd}, tracks the desired output trajectory $y^d(t)$ in an 
exponentially stable manner if $\hat\cF_k$ converges to $\cF_k$ in finite time.
\end{proof}
%\begin{remark}
%Although the main results in this paper have been developed for output tracking control, they can be 
%applied to state tracking control for system states that are observable from the outputs measured. 
%For state tracking, the time dependence of these states on the inputs and outputs must be known 
%for the system as given by eq. \eqref{dynoutmod}.
%\end{remark}

The above result, in combination with the output observer in Section \ref{sec: ftsobs} and the two ultra-local 
model observers in Section \ref{sec: ulm}, is applied to a second-order system, the inverted pendulum on a 
cart with nonlinear friction terms, in numerical simulations carried out in the following section. 

\section{Numerical Simulation Results}\label{sec: numres}
In this section, we provide numerical simulation results of the model-free tracking control framework on an 
inverted pendulum on a cart with nonlinear friction terms affecting the motion of both the degrees of freedom. 
The dynamics model of this system is unknown to the controller. This system is described in Section 
\ref{ssec: invpen} and the numerical results of the control scheme are given in Section \ref{ssec: simres}.
\subsection{Inverted pendulum on cart system}\label{ssec: invpen}
\begin{figure}[hbt]
\begin{center}
\hspace*{-2mm}\includegraphics[height=2.3in]{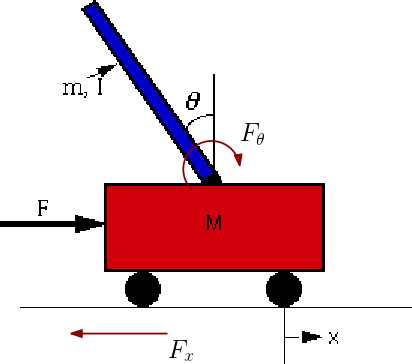}
\caption{Inverted pendulum system to which our nonlinear model-free control framework 
is applied.}  % width is 8.4 cm.
\label{invpenfig}                                 % Size the figures 
\end{center}        
\vspace*{-2mm}                         % accordingly.
\end{figure}
The inverted pendulum on cart is a two degree-of-freedom mechanical system, with the cart position $x$ 
considered positive to the right of an inertially-fixed origin and the angular displacement $\theta$ considered 
positive counter-clockwise from the upward vertical, as shown in Fig. \ref{invpenfig}. The input to the 
system is a horizontal force on the cart denoted $F$ in this figure, and the output is the angular 
displacement of the pendulum $\theta$; therefore, this is a single input single output (SISO) system. The 
mass and rotational inertia of the pendulum are $m$ and $I$ respectively, its length is $2l$, and the mass 
of the cart is $M$. A dynamics model of the system, which is unknown for the purpose of control design, 
is used to generate the desired output trajectory to be tracked. Then the model-free control scheme is 
used to track this desired trajectory.  

For simulation purposes, the inverted pendulum on a cart system is subjected to a nonlinear friction 
force acting on the cart's motion, and a nonlinear friction-induced torque acting on the pendulum. The 
friction force acting on the cart is denoted $F_x$ and the friction torque acting on the pendulum is 
denoted $F_\theta$, and they are given by:
\be F_x= c_x\tanh\dot x,\; F_\theta= c_\theta\tanh\dot\theta. \label{fricts} \ee
Note that the hyperbolic tangent function ensures that these frictional effects get saturated at high 
speeds ($\dot  x$ and $\dot\theta$). Therefore, the dynamics model of this system, which is unknown 
for the purpose of control design, is given by:
\begin{align}
\begin{split}
&\cM(q)\ddot q+ \cD(q,\dot q)= b F,\; q= \bbm x\\ \theta \ebm,\, b= \bbm 1\\ 0\ebm, \\
&\cM(q)= \bbm M+m & -ml\cos\theta \\ -ml\cos\theta & I+ ml^2\ebm, \\ &\cD(q,\dot q)= 
\bbm ml\dot\theta^2 \sin\theta+ c_x\tanh\dot x \\ c_\theta\tanh\dot\theta - mgl\sin\theta \ebm.
\end{split} \label{invpenmodel}
\end{align}
The input and output are:
\be u= F,\;\ y=\theta. \label{inpout-invpen} \ee
For the purpose of the numerical simulation, the parameter values selected for this system are:
\begin{align}
\begin{split}
&M=1.5\, \mrm{ kg}, \; m=0.5\, \mrm{ kg},\; l=1.4\, \mrm{ m},\; I=0.84\, \mrm{ kg\, m}^2,\\
&g= 9.8\, \mrm{ m/s}^2,\; c_x=0.028\, \mrm{ N}, \; c_\theta=0.0032\, \mrm{ N\, m}. 
\end{split} \label{sysparams} 
\end{align}
The desired trajectory was generated by applying the following open-loop control input (force) to 
the cart:
\be F= -c_x \dot x -0.5 c_\theta \dot\theta - 0.1 c_x x.
%ml\dot\theta^2\sin\theta - (M+m)g\sin\theta- 2 (M+m\sin^2\theta) g\sin\theta
\label{OLcont} \ee
This generates an output trajectory $\theta^d (t)$ that is oscillatory with slowly decreasing amplitude, 
as depicted in Fig. \ref{gentraj} in Section \ref{ssec: simres}. Note that the model used here is purely 
for the purpose of trajectory generation and to demonstrate the working of the model-free control 
framework outlined in this paper. The framework itself is more widely applicable to systems that may 
not have known input-output (or input-state) models or systems that are very difficult to model, e.g., 
biological processes. 

\subsection{Simulation results of control scheme}\label{ssec: simres}
Here we present numerical simulation results for the model-free tracking control scheme applied 
to the system described by eqs. \eqref{invpenmodel}-\eqref{sysparams}. A trajectory is generated 
for this system using the control scheme \eqref{OLcont} along with the initial states: 
\be \bbm q^d(0)\\ \dot q^d(0) \ebm= \bbm x^d(0)\\ \theta^d (0)\\ \dot x^d(0)\\ \dot\theta^d(0) \ebm= 
\bbm 0.45\, \mbox{m}\\ -0.14\, \mbox{rad} \\ -0.3\, \mbox{m/s} \\ 0.05\, \mbox{rad/s} \ebm. \label{initstates} \ee
The generated trajectory $y^d(t)=\theta^d (t)$ for a time interval of $T=70$ seconds is depicted in 
Fig. \ref{gentraj}. 
\begin{figure}[htb]
\begin{center}
\includegraphics[width=\columnwidth]{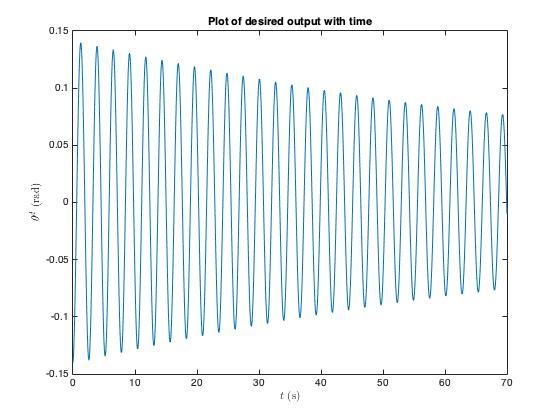}
\caption{Desired trajectory generated for $T=70$ seconds for inverted pendulum on cart system.}  
% width is 8.4 cm.
\label{gentraj}                                 % Size the figures 
\end{center}        
\vspace*{-2mm}                         % accordingly.
\end{figure}

The control scheme given by Corollary \ref{cont2nd} is then applied to this system to track this  
desired trajectory. For this simulation, we assume that the initial estimated states are: 
\be \bbm \hat q(0)\\ \dot{\hat q}(0) \ebm= \bbm \hat x(0)\\ \hat \theta(0)\\ \dot{\hat x}(0)\\ \dot{\hat\theta}(0) 
\ebm= \bbm 0\, \mbox{m}\\ 0.102\, \mbox{rad} \\ 0\, \mbox{m/s} \\ 0\, \mbox{rad/s} \ebm. \label{initcs} \ee
Measurements of the output are assumed at a constant rate of 50 Hz, i.e., sampling period $\Delt= 0.02$ s. 
In the simulation, the measurements are generated by numerically propagating the true dynamics model 
of the inverted pendulum on cart system given earlier, and adding noise to the output $y(t)=\theta(t)$. 
The additive noise is generated by a random number generator that uses a bump function of width 
$0.018$ rad ($\approx 1.03^\circ$) as a probability distribution function. Observer gains used for this 
simulation, with the observer structure given in Theorem \ref{dFTSobs}, are:
\be L =2.1,\;\ \beta=2,\; \mbox{ and }\; p=\frac75. \label{obsv-gains} \ee
%A discretized second-order Butterworth filter is used to estimate and predict the ultra-local model from 
%past input-output history, as outlined in Section \ref{ssec: estimdyn}. This second-order Butterworth filter 
%has a damping constant of $\zeta=1/\sqrt{2}$ and a natural frequency of $\omega_n=1$ rad/s. 
The first order ultra-local model observer given by Proposition \ref{prop1st} is used, with observer gains:
\be \lambda=1.5,\; \mbox{ and }\; r=\frac97. \label{ulmobs-gains} \ee
This observer is initialized with the zero vector, i.e., $\hat\cF_0=0$.
%As mentioned in Remark 4 at the end of Section \ref{ssec: estimdyn}, we initialize the control scheme 
%using the control law \eqref{MFClaw2nd} along with the assumptions that $e_{-1}=0$ and $\cF_0^{-}=0$.
%input at the first time step given by the simple proportional feedback law: 
%\be u_0=  0.7 (y^d(t)- \hat y (t))/G_0. \label{u0calc} \ee
The control law \eqref{MFClaw2nd} is then used to compute the control inputs $u_k$ for $k>1$. The control 
gains used in this simulation are:
\begin{align} 
\begin{split}
&\eta=1,\; p=\frac{11}{9}, \; \mu=0.35, \\
&\mbox{and }\; \cG_k= 1.5 \big(1+\tanh (\|E_k\|)\big), %\frac{1.4}{1+\exp(-\|E_k\|)}, 
\end{split}\label{contr-gains} 
\end{align}
where $E_k$ is the total of the last three terms in eq. \eqref{intMFC2nd}.

\begin{figure}[htb]
\begin{center}
\includegraphics[width=\columnwidth]{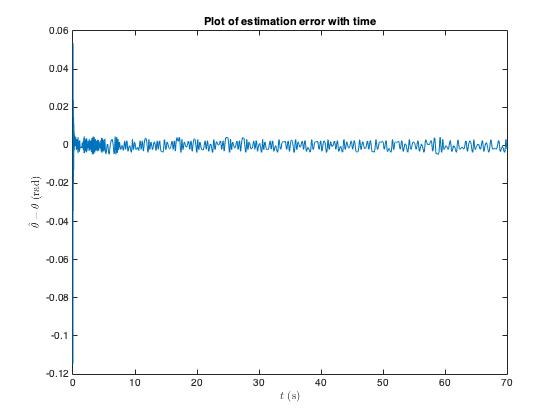}
\vspace*{1mm}
\includegraphics[width=\columnwidth]{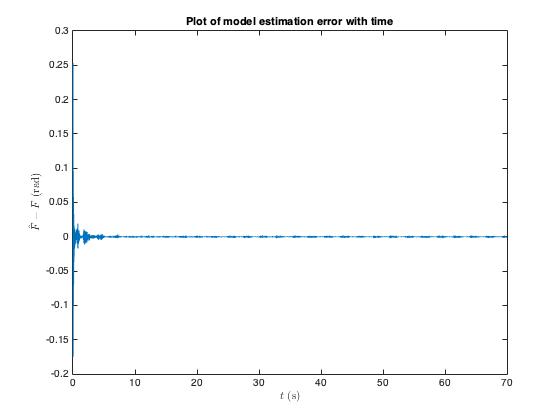}
\vspace*{-5mm}
\caption{Estimation errors in output estimation (top) and ultra-local model estimation (bottom) for inverted 
pendulum on cart system with model-free control.}  % width is 8.4 cm.
\label{estimerrs}                                 % Size the figures 
\end{center}        
%\vspace*{-2mm}                         % accordingly.
\end{figure}
The simulation results for estimation error in estimating the output from noisy measurements ($e^0_k$) 
using the finite-time stable observer outlined in Section \ref{ssec: ftsobss}, and estimation error in estimating 
the ultra-local model according to Section \ref{ssec: estF1st} are depicted in Fig. \ref{estimerrs}. Simulation 
results for the tracking control performance are shown in Fig. \ref{tracontinp}. 
\begin{figure}[hbt]
\begin{center}
\includegraphics[width=\columnwidth]{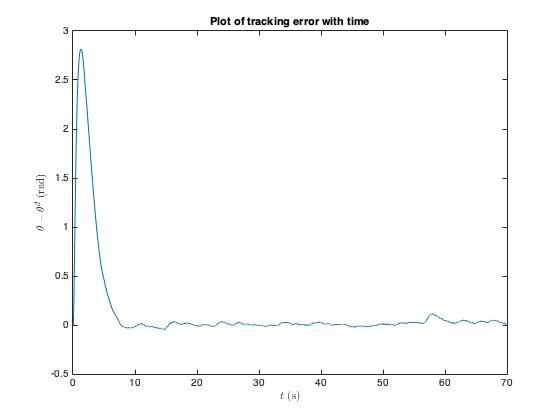}
\vspace*{1mm}
\includegraphics[width=\columnwidth]{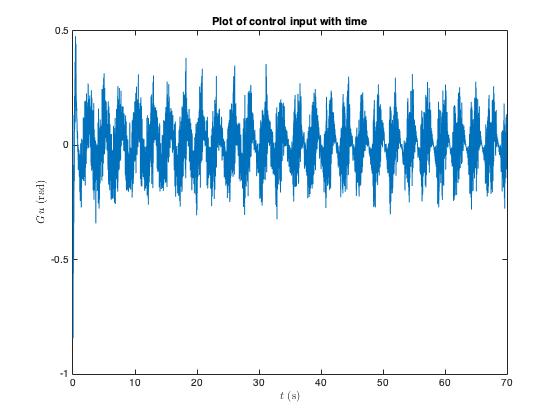}
\vspace*{-5mm}
\caption{Output trajectory tracking error (top) and control input (bottom) for inverted pendulum on cart system 
with model-free control.}  % width is 8.4 cm.
\label{tracontinp}                                 % Size the figures 
\end{center}        
%\vspace*{-2mm}                         % accordingly.
\end{figure}
The plot on the top shows the output trajectory tracking error over the simulated duration. Note that the 
tracking error settles down to within an error bound less than about $0.3$ rad in steady state after an 
initial brief period of transients. The time plot of the control input is shown in the bottom plot.
%Fig. \ref{continput}. 
%\begin{figure}[htb]
%\begin{center}
%\includegraphics[width=\columnwidth]{continput.jpg}
%\vspace*{-2mm}   
%\caption{Time plot of the control input for inverted pendulum on cart system using the model-free control 
%law eq. \eqref{intMFC2nd}.}  
%\label{continput}                                 % Size the figures 
%\end{center}        
%\vspace*{-2mm}                         % accordingly.
%\end{figure}
This control input profile shows some high frequency oscillations in tracking the desired trajectory, that 
seem to correlate with the oscillations seen in the output observer error (due to measurement noise) and 
therefore the ultra-local model observer error in Fig. \ref{estimerrs}. Future work will deal 
with reducing these transients by doing one or more of the following: (i) using more advanced schemes 
for predicting the ultra-local model from past input-output history; (ii) tuning of observer and controller  
gains to reduce the amplitude of oscillations; and (iii) using integral term(s) in the observer designs to 
produce smoother estimates of the output and ultra-local model. A reference governor may also be used 
to modify the reference (desired) output trajectory based on current estimates of outputs as in, 
e.g.,~\cite{ergov16}.  
%(i) better initialization methods for   
%the control scheme to smoothen the initial transient control effort; (ii) using a reference governor to 
%modify the reference (desired) output trajectory based on current estimates of outputs as in, 
%e.g.,~\cite{ergov16};  and (iii) using more advanced schemes for predicting the ultra-local model from 
%past input-output history.  

\begin{remark}
Although the schemes given here assume that the output space is a vector space, the angle output 
for this inverted pendulum on cart example is on the circle $\bS^1$, which is not a vector space. 
Therefore, the observer and control laws outlined in the earlier sections may lead to unwinding, 
even though that does not happen for the numerical simulation reported here. The model-free observer 
and controller design framework outlined here will be extended to systems evolving on non-Euclidean 
output (or state) spaces in the future, to address this issue.
\end{remark}

\section{Conclusion}\label{sec: conc}
This paper presents a formulation of a model-free control approach that guarantees nonlinear stability 
for output tracking control with feedback of output measurements that may contain additive noise. The 
formulation presented here is developed in discrete time, and uses the concept of a control affine  
ultra-local model used to model unknown input-output behavior that was used in the linear model-free 
control approach formulated by Fliess and Join in the last decade. However, that is where the 
similarity ends. The first part of the framework given here uses a continuous nonlinear observer for 
estimating the outputs from the measurements. This observer ensures finite-time stable convergence 
of the output estimation errors to zero, which in turn enables separate design of a continuous nonlinear 
controller for output feedback tracking. The second part of the framework develops nonlinearly stable 
and robust observers to estimate the ultra-local model that models the unknown input-output dynamics, 
from past input-output history. In the last part of the framework, a nonlinear output feedback tracking 
control law is designed that uses estimates of the measured output and the ultra-local model, to give a 
nonlinearly stable and robust control scheme. Nonlinear stability analysis shows the stability of the 
feedback compensator combining the nonlinear observers and nonlinear control law when the change 
in the discrete-time system dynamics modeled by the ultra-local model has a bounded finite difference. 
A numerical simulation experiment is carried out on an inverted pendulum on a cart system with 
nonlinear friction, for which the input is the horizontal force applied to the cart and the output is the angle 
from the upward vertical of the pendulum. Noisy measurements of the output are available with bounded 
amplitude of noise. The model of the dynamics of this system is unknown to the nonlinear observer and 
controller designed using our nonlinear model-free control framework. This numerical experiment shows 
convergence of output estimation errors and output tracking errors to small absolute values. Future work 
will explore extensions of this framework to systems evolving on Lie groups and their principal bundles, 
and also development of stable higher-order observers for the ultra-local model for increased robustness.

\section{Acknowledgements}
A large portion of this work was carried out by the author when he was hosted by the Systems and 
Control (SysCon) department at Indian Institute of Technology, Bombay, India, (IIT-B), in the summer of 
2019. Helpful discussions with his hosts at SysCon (IIT-B), Debasish Chatterjee and Sukumar Srikant, 
are gratefully acknowledged. 

\bibliographystyle{IEEEtran}
\bibliography{alias,references}

%\section*{Appendix}
%Proof(s) of stability and figures of simulation results may go here.

\end{document}